\documentclass[3p,authoryear,12pt]{elsarticle}

\usepackage{graphicx}
\usepackage{amsmath,amssymb,amsfonts,longtable,rotating,color,subfigure}

\newtheorem{theorem}{Theorem}[section]

\newtheorem{lemma}[theorem]{Lemma}

\newenvironment{proof}{\begin{trivlist} \item[] {\bf Proof:}}{
\hspace*{\fill}$\Box$ \end{trivlist}}
\usepackage{epstopdf}
\usepackage{tikz}
\usetikzlibrary{decorations.pathreplacing}
\usepackage{url}
\usepackage{float} % I have added this to palce the tables
\usepackage{blkarray} % I have added this for matrix
\usepackage{multirow}
\makeatletter
\def\ps@pprintTitle{%
 \let\@oddhead\@empty
 \let\@evenhead\@empty
 \def\@oddfoot{\centerline{\thepage}}%
 \let\@evenfoot\@oddfoot}
\makeatother

\journal{Communications in Statistics - Theory and Methods }

\begin{document}

\begin{frontmatter}

\title{A  monotonicity property of weighted log-rank tests}

\author{Tahani Coolen-Maturi \corref{cor}}
\ead{tahani.maturi@durham.ac.uk}
\cortext[cor]{Corresponding author}
\author{Frank P.A.\ Coolen}
\ead{frank.coolen@durham.ac.uk}
\address{Department of Mathematical Sciences, Durham University, Durham, DH1 3LE, UK}

\begin{abstract}
The logrank test is a well-known nonparametric test which is often used to compare the survival 
distributions of two samples including right censored observations, it is also known as the 
Mantel-Haenszel test. The $G^{\rho}$ family of tests, introduced by \citet{HF82}, generalizes the
logrank test by using weights assigned to observations. In this paper, we present a monotonicity
property for the $G^{\rho}$ family of tests, which was motivated by the need to derive bounds for
the test statistic in case of imprecise data observations. 
\end{abstract}

\begin{keyword}
Logrank, Monotonicity, Imprecise probability, survival distribution, monotonicity property .
\end{keyword}
\end{frontmatter}

\section{Introduction}\label{intro}

The logrank test is a well-known nonparametric test which is often used to compare  the survival distributions of two  samples
containing right censored observations. It generalizes the Wilcoxon test, for data without right censored observations, and
is also known as Mantel-Haenszel test  \citep{MN66}. Several variations to this test have been introduced in the literature, 
e.g.\ Gehan's Generalized Wilcoxon test \citep{Gehan65,LL98}, Weighted Logrank tests \citep{LB77} and Peto-Peto  
test \citep{Peto72}. The Mantel-Haenszel test  \citep{MN66} gives equal weights to observations regardless of the time 
at which an event occurs. On the other hand, the Peto-Peto   test statistic assigns more weights to earlier event 
times \citep{Peto72}. \citet{HF82} introduced a class of tests, the $G^{\rho}$ family, which can be used to test the 
null hypothesis $H_0: S_0(t)=S_1(t)$ for all  $t>0$ against the alternative hypothesis $H_1: S_0(t)\neq S_1(t)$ for some $t>0$.

In this paper, we consider the $G^{\rho}$ family of tests for right censored data introduced by \citet{HF82}, in 
which the weight assigned to each observed failure time $t$ is of the form $[\hat{S}(t)]^{\rho}$ for fixed $\rho\geq 0$, 
where $\hat{S}(t)$ is the well-known Kaplan-Meier estimate of the survival function. Note that the use of `failure time' does
not restrict the test applications and could be interpreted as time of any event of interest, as long as each individual (or `item')
has only one event associated with it, which may either be observed (failure time) or only known to be greater than an observed
right censoring time. Throughout this paper it is assumed that right censoring is non-informative.
As special cases, $\rho=0$ gives the log-rank Mantel-Haenszel test \citep{MN66} and  $\rho=1$ gives the Peto-Prentice extension 
of the Wilcoxon statistic when $\rho=1$  \citep{Peto72,Prentice78}.
Several R packages are available  to perform these tests, e.g.\ \texttt{survdiff} within the \textbf{survival} package 
\citep{Rsurvival}  and the comprehensive \textbf{FHtest} package \citep{RFHtest}.

In this paper we prove a monotonicity property of  the $G^{\rho}$ class family  of tests for right censored data 
introduced by \citet{HF82}. Formally, a function $f$ is called monotonically non-decreasing if it preserves the order, 
that is if for all $a$ and $b$ with $a\leq b$ we have $f(a)\leq f(b)$. This research was motivated by possible applications
of such tests in case of imprecise data, where the ordering of observations per group is known but where the ranking
of observations between the groups may not be precisely determined due to lack of precise values for some or all of the
observations, it is most natural to assume that each observation is only known to belong to an interval. In such
cases, when intervals are overlapping, different combined rankings of the data 
from different groups may be possible and one typically would like to find the
minimum and maximum values of the test statistic corresponding to all possible
combined rankings. The result in this paper makes the derivation of these minimum
and maximum values straightforward. It should be noticed that such monotonicity 
trivially holds for the test statistic of the Wilcoxon rank-sum test, so the main challenge
here results from the presence of right-censored observations in the data set.

This paper is organised as follows. Section \ref{sec.notation} introduces the notation and the setting, while the main results are presented in Section \ref{sec.main}. An example is provided in Section \ref{sec.example}, and finally  the paper ends with concluding remarks in Section \ref{sec.con}.

%%%%%%%%%%%%%%%%%%%%%%%%%%%%%%%%

\section{Notation and Setting}\label{sec.notation} 
 Let $\tau_1<\tau_2<\ldots<\tau_k$ denote $k$ times of observed failures.  Let $Y_i(\tau_j)$ be  the number of individuals in group $i$ who are at risk at $\tau_j$ $(i = 0,1)$, i.e.\  the  number of individuals from both groups at risk at $\tau_j$  is $Y (\tau_j) = Y_0(\tau_j) + Y_1(\tau_j) $, $j = 1,2,\dots,k$. Let $d_{ij}$ be the number of individuals in group $i$ who fail  at $\tau_j$ $(i = 0,1)$, so  the total number of failures at $\tau_j$ from both groups is $d_j = d_{0j} +d_{1j}$, $j = 1,2,\dots,k$. The information at time $\tau_j$ can be summarised in the following $2\times2$ table:

%\begin{table}[htp]
%\caption{default}
\begin{center}
\begin{tabular}{cccc}
&fail at $\tau_j$ &censored&at risk at $\tau_j$\\
\cline{2-3}
Group 0 & \multicolumn{1}{|c|}{$d_{0j}$}  & \multicolumn{1}{|c|}{$Y_0(\tau_j) - d_{0j}$} & $Y_0(\tau_j)$\\
\cline{2-3}
Group 1 & \multicolumn{1}{|c|}{$d_{1j}$}  & \multicolumn{1}{|c|}{$Y_1(\tau_j) - d_{1j}$} & $Y_1(\tau_j)$\\
\cline{2-3}
 & $d_{j}$  & $Y(\tau_j) - d_{j}$ & $Y(\tau_j)$\\
\end{tabular}
\end{center}

%Under the null  hypothesis, the following statistic follows asymptotically a standard normal distribution, that is
Consider the test statistic
\begin{equation}\label{eq.z}
Z=\frac{O-E}{\sqrt{V}}=\frac{\sum_j O_j- \sum_j E_j}{\sqrt{\sum_jV_j}}
\end{equation}
with 
\begin{align}
O_j &= \left[\hat{S}(\tau_j)\right]^{\rho} d_{1j}\\
E_j&=\left[\hat{S}(\tau_j)\right]^{\rho}\left(\frac{Y_1(\tau_j)}{Y(\tau_j)}\right) d_j \label{eq.exp1}\\ 
V_j&=\left[\hat{S}(\tau_j)\right]^{2\rho}\frac{Y_0(\tau_j)Y_1(\tau_j)}{(Y(\tau_j))^2}Y_1(\tau_j) \left(\frac{d_j}{Y(\tau_j)}\right)  \left(1-\frac{d_j}{Y(\tau_j)}\right)\label{eq.var1}
\end{align}
where $\rho\geq 0$  and $\hat{S}(\tau_j)$ is the  Kaplan-Meier  estimate at time $\tau_j$ \citep{KM58}. Then under  the null  hypothesis $H_0: S_0(t)=S_1(t)$ for all  $t>0$,  the test statistic $Z$ follows the standard normal distribution, i.e. $Z \sim N(0,1)$, so $Z^2 \sim \chi^2_1$.\\

For simplicity of notation, we assume throughout this paper that there are no ties, therefore $d_j=d_{0j}+d_{1j}=1$, and $O=\sum_j \left[\hat{S}(\tau_j)\right]^{\rho} d_{1j}$ is the number of failures from group $G_1$.  The expected value formula \eqref{eq.exp1} and the variance formula \eqref{eq.var1} can be simplified (as $d_j=1$) as
\begin{align}
E_j&=\left[\hat{S}(\tau_j)\right]^{\rho}\frac{Y_1(\tau_j)}{Y(\tau_j)}\label{eq.exp2}\\
V_j&=\left[\hat{S}(\tau_j)\right]^{2\rho}\frac{Y_0(\tau_j)Y_1(\tau_j)}{(Y(\tau_j))^2} \label{eq.var2}
\end{align}

 Now let $\tau_{k_j-1} \leq x_{j_i} < y_j \leq \tau_{k_j} $,  and let $u_0$ ($u_1$) be the number of censored observations from group $G_0$ ($G_1$) between $\tau_{k_j-1}$ and  $ \tau_{k_j} $, thus $u=u_0+u_1$. Define $\delta_i(\tau_{j})$ to be equal  to 1 if  $\tau_{j}$  is a failure from group $G_i$, and zero otherwise, $i=0,1$. Let $Y_i(\tau_{k_j})$ be the number of individuals in group $G_i$ who are at risk at $\tau_{k_j}$, $i=0,1$, and  let  $Y(\tau_{k_j})=Y_0(\tau_{k_j})+Y_1(\tau_{k_j})$ be the number of individuals from both groups at risk at $\tau_{k_j}$, $k_j\in \{1,2,\ldots,k\}$.
 This is illustrated in the first row of  Figure \ref{plot0}. The next section introduces the main results of this paper.\\[.5ex]

\section{Main Results}\label{sec.main}

In this section, we consider the following setting. For  a particular data set, with fixed failure-censored status, suppose that  all observations from group $G_0$, $x_1<x_2<\ldots <x_{n_0}$,  precede  all observations from group $G_1$, $y_1<y_2<\ldots <y_{n_1}$. We would like to swap  between neighbouring data observations, one pair of a $G_0$ observation and a $G_1$ observation at a time, where the
latter is the smallest $G_1$ observation greater than the $G_0$ observation, until we have all observations from group $G_1$  preceding  all observations from group $G_0$. In total we can do that in $n_0 n_1$ steps (switches), where $n_0$  and $n_1$  are the sample sizes of  group $G_0$ and group $G_1$, respectively. For example, if we have 3 observations from each group, the number of switches from $x_1<x_2<x_3<y_1<y_2<y_3$ to $y_1<y_2<y_3<x_1<x_2<x_3$ is 9. The property presented in this paper is that, under the null hypothesis,  the $z$-test statistic behaves monotonically throughout this swapping process.  To explain this clearly, we need to introduce further notation. Let  $Z_B$  be the $z$-test statistic value corresponding to the permutation {\it before} swapping the adjacent values  $ x_{j_i} < y_j$, and $Z_A$  is the value of the $z$-statistic  {\it after} the swap. That is we want to show that  $Z_B\leq Z_{A}$ for the  $n_0 n_1$ switches. This is stated in the following theorem.

%The idea is that we would like to proof that, under the null hypothesis,  the $Z$-test statistic is monotone that is $Z_i<Z_{i+1}$, $\forall i\in\{1,2,\ldots,n_0 n_1+1\}$. This is stated in the following theorem. \\

\begin{theorem}
Under the null  hypothesis,  the two-sample  logrank test statistic $Z=\frac{O-E}{\sqrt{V}}$, given in \eqref{eq.z}, is  monotonic.
\end{theorem}

%Below we provide the proof for the above theorem.\\

The remainder of this section consists of the proof of this theorem. To start the proof, suppose  we swap $x_{j_i} $ and $y_j$, that is now $x_{j_i} > y_j$, then we have four different scenarios we need to consider:\\[.5ex]

\begin{figure}[ht]
\begin{center}
\resizebox{!}{6cm}{%
\begin{tikzpicture}
%bottom right
\draw[step=1cm,very thin] (5,0) grid (7,2);
\node at (6,2.3) {$\tau_{k_j}$};
\node at (10.1,1.5) {$Y_0(\tau_{k_j})=Y_0(\tau_{k_j-1})-u_0-\delta_0(\tau_{k_j-1})$};
\node at (10.1,.5) {$Y_1(\tau_{k_j})=Y_1(\tau_{k_j-1})-u_1-\delta_1(\tau_{k_j-1})$};
\node at (9.4,-.5) {$Y(\tau_{k_j})=Y(\tau_{k_j-1})-u-1$};

% bottom left

\draw[step=1cm,very thin] (0,0) grid (2,2);
\node at (1,2.3) {$\tau_{k_j-1}$};
\node at (2.9,1.5) {$Y_0(\tau_{k_j-1})$};
\node at (2.9, .5) {$Y_1(\tau_{k_j-1})$};
\node at (2.9,-.5) {$Y(\tau_{k_j-1})$};

% top left
\draw[step=1cm,very thin] (0,4) grid (2,6);
\node at (1,6.3) {$\tau_{k_j-1}$};
\node at (2.9,5.5) {$Y_0(\tau_{k_j-1})$};
\node at (2.9, 4.5) {$Y_1(\tau_{k_j-1})$};
\node at (2.9,3.5) {$Y(\tau_{k_j-1})$};

% top right
\draw[step=1cm,very thin] (5,4) grid (7,6);
\node at (6,6.3) {$\tau_{k_j}$};
\node at (10.1,5.5) {$Y_0(\tau_{k_j})=Y_0(\tau_{k_j-1})-u_0-\delta_0(\tau_{k_j-1})$};
\node at (10.1,4.5) {$Y_1(\tau_{k_j})=Y_1(\tau_{k_j-1})-u_1-\delta_1(\tau_{k_j-1})$};
\node at (9.4,3.5) {$Y(\tau_{k_j})=Y(\tau_{k_j-1})-u-1$};

\node at (-2,1) {After};
\node at (-.5,1.5) {$G_0$};
\node at (-.5, .5) {$G_1$};

\node at (-2,5) {Before};
\node at (-.5,5.5) {$G_0$};
\node at (-.5, 4.5) {$G_1$};

\end{tikzpicture}
}
\caption{Setting and Scenario S1}
\label{plot0}
\end{center}
\end{figure}
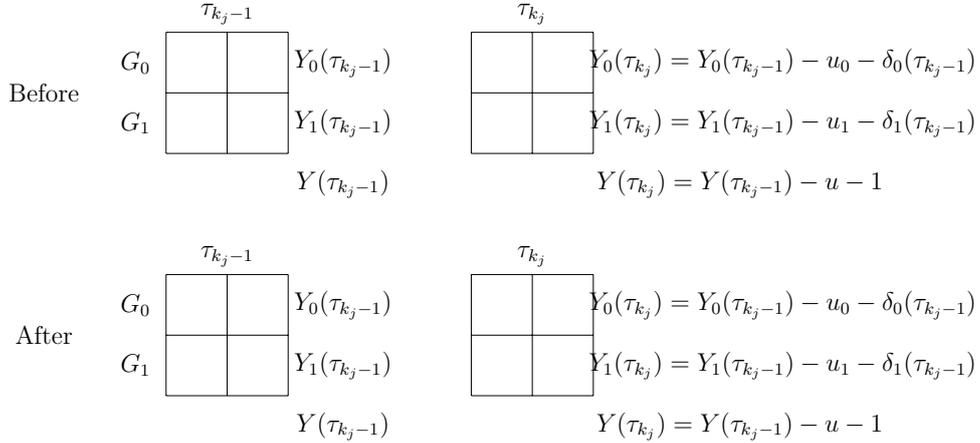

\noindent  {\bf Scenario 1 (S1):  when both $x_{j_i}$ and $y_j$ are censoring times}\\
In this case, nothing will change to the  $2\times 2$ tables in Figure \ref{plot0}, where the first row is corresponding to  {\it before} the swap and the second row to {\it after} the swap. As the value of  $\hat{S}(\tau_j)$ is a step function that change value only at the time of observed failure, therefore  the expected value and the variance formula are the same before and after the swap. That is if we swap any two censored observations  between $\tau_{k_j-1}$ and  $ \tau_{k_j} $ this will not affect the expected value and the variance, as it does not affect the margins in the $2\times 2$ tables in Figure \ref{plot0}. Thus 
$Z_B= \frac{(O-E_B)}{\sqrt{V_B}}$ and $Z_A= \frac{(O-E_A)}{\sqrt{V_A}}$ are equal, $Z_A=Z_B$, where we use $B$ as subscript for the case before the swap and $A$ for the case after the swap.\\

The proofs for the next three cases are very similar, yet for the sake of completeness full details are given. \\

\noindent  {\bf Scenario 2 (S2):  when $x_{j_i}$ is a failure time and  $y_j$ is a censoring time}\\
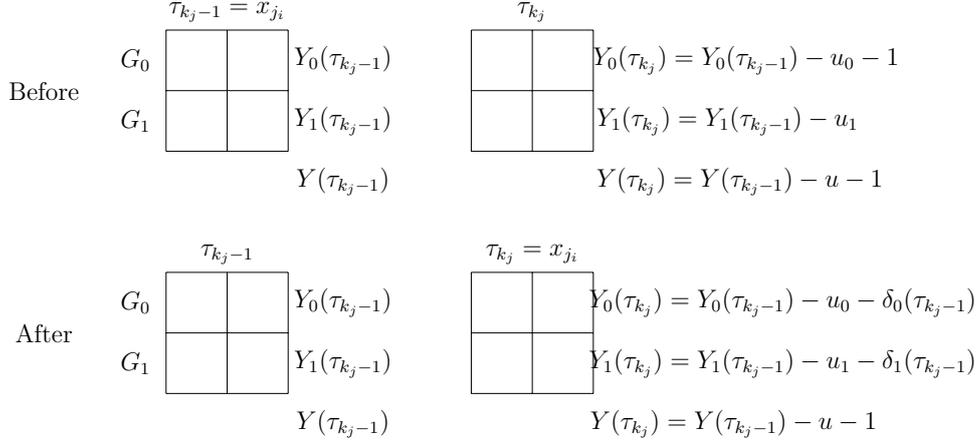
\begin{figure}[ht]
\begin{center}
\resizebox{!}{6cm}{%
\begin{tikzpicture}
%bottom right
\draw[step=1cm,very thin] (5,0) grid (7,2);
\node at (6,2.3) {$\tau_{k_j}=x_{j_i}$};
\node at (10.1,1.5) {$Y_0(\tau_{k_j})=Y_0(\tau_{k_j-1})-u_0-\delta_0(\tau_{k_j-1})$};
\node at (10.1,.5) {$Y_1(\tau_{k_j})=Y_1(\tau_{k_j-1})-u_1-\delta_1(\tau_{k_j-1})$};
\node at (9.3,-.5) {$Y(\tau_{k_j})=Y(\tau_{k_j-1})-u-1$};

% bottom left

\draw[step=1cm,very thin] (0,0) grid (2,2);
\node at (1,2.3) {$\tau_{k_j-1}$};
\node at (2.9,1.5) {$Y_0(\tau_{k_j-1})$};
\node at (2.9, .5) {$Y_1(\tau_{k_j-1})$};
\node at (2.9,-.5) {$Y(\tau_{k_j-1})$};

% top left
\draw[step=1cm,very thin] (0,4) grid (2,6);
\node at (1,6.3) {$\tau_{k_j-1}=x_{j_i}$};
\node at (2.9,5.5) {$Y_0(\tau_{k_j-1})$};
\node at (2.9, 4.5) {$Y_1(\tau_{k_j-1})$};
\node at (2.9,3.5) {$Y(\tau_{k_j-1})$};

% top right
\draw[step=1cm,very thin] (5,4) grid (7,6);
\node at (6,6.3) {$\tau_{k_j}$};
\node at (9.5,5.5) {$Y_0(\tau_{k_j})=Y_0(\tau_{k_j-1})-u_0-1$};
\node at (9.2,4.5) {$Y_1(\tau_{k_j})=Y_1(\tau_{k_j-1})-u_1$};
\node at (9.4,3.5) {$Y(\tau_{k_j})=Y(\tau_{k_j-1})-u-1$};

\node at (-2,1) {After};
\node at (-.5,1.5) {$G_0$};
\node at (-.5, .5) {$G_1$};

\node at (-2,5) {Before};
\node at (-.5,5.5) {$G_0$};
\node at (-.5, 4.5) {$G_1$};

\end{tikzpicture}
}
\caption{Scenario S2}
\label{plot1}
\end{center}
\end{figure} 

This second scenario is illustrated in Figure \ref{plot1}, and the  expect values for before and after the swap are given as

\[E_B= \Sigma_{j\neq k_j} E_j+ \left[\hat{S}(\tau_j)\right]^{\rho} \frac{Y_1(\tau_{k_j-1})-u_1}{Y(\tau_{k_j-1})-u-1}\] % before

\[E_A= \Sigma_{j\neq k_j} E_j+\left[\hat{S}(\tau_j)\right]^{\rho} \frac{Y_1(\tau_{k_j-1})-u_1-\delta_1(\tau_{k_j-1})}{Y(\tau_{k_j-1})-u-1} \] % after

Clearly $E_B \geq E_A$ thus $(O-E_B) \leq (O-E_A)$. And  the variances 

\[V_B= \Sigma_{j\neq k_j} V_j+\left[\hat{S}(\tau_j)\right]^{2\rho} \frac{(Y_0(\tau_{k_j-1})-u_0-1)(Y_1(\tau_{k_j-1})-u_1)}{(Y(\tau_{k_j-1})-u-1)^2}\] % before

\[V_A= \Sigma_{j\neq k_j} V_j+\left[\hat{S}(\tau_j)\right]^{2\rho} \frac{(Y_0(\tau_{k_j-1})-u_0-\delta_0(\tau_{k_j-1}))(Y_1(\tau_{k_j-1})-u_1-\delta_1(\tau_{k_j-1}))}{(Y(\tau_{k_j-1})-u-1)^2} \] % after

\vspace{.5cm}

We have two main cases:\\
\begin{itemize}
\item[(i)] If $\delta_0(\tau_{k_j-1})=1$, that is $\tau_{k_j-1}$  is a failure from group $G_0$ and thus $\delta_1(\tau_{k_j-1})=0$, then $V_B=V_A$. Thus $Z_B \leq Z_A$, as from above $(O-E_B) \leq (O-E_A)$.\\

\item[(ii)] If  $\delta_1(\tau_{k_j-1})=1$, that is $\tau_{k_j-1}$  is a failure from group $G_1$ and thus $\delta_0(\tau_{k_j-1})=0$, then

\[V_A= \Sigma_{j\neq k_j} V_j+\left[\hat{S}(\tau_j)\right]^{2\rho} \frac{(Y_0(\tau_{k_j-1})-u_0)(Y_1(\tau_{k_j-1})-u_1-1)}{(Y(\tau_{k_j-1})-u-1)^2} \] % after
 and thus we have two sub-cases:\\
 
 \begin{itemize}
 \item[(iia)] If $Y_1(\tau_{k_j-1})-u_1>Y_0(\tau_{k_j-1})-u_0$ then $V_B<V_A$,  i.e. $\frac{1}{V_B}>\frac{1}{V_A}$.
  \begin{itemize}
\item[-] If $(O-E_B)>0$, then $(O-E_A)$ also has to be positive.\\
We multiply $(O-E_B) \leq (O-E_A)$ by $\frac{1}{\sqrt{V_B}}$ and  by $\frac{1}{\sqrt{V_A}}$, then we have
\begin{align*}
\frac{(O-E_B)}{\sqrt{V_A}} &\leq \frac{(O-E_A)}{\sqrt{V_A}}   \\
 \frac{(O-E_B)}{\sqrt{V_B}} &\leq  \frac{(O-E_A)}{\sqrt{V_B}} 
\end{align*}
Now we multiply $\frac{1}{\sqrt{V_A}}<\frac{1}{\sqrt{V_B}}$ by $(O-E_A)$ and  by $(O-E_B)$, then we have
\begin{align*}
 \frac{(O-E_A)}{\sqrt{V_A}} &< \frac{(O-E_A)}{\sqrt{V_B}}  \\
\frac{(O-E_B)}{\sqrt{V_A}} &< \frac{(O-E_B)}{\sqrt{V_B}}
\end{align*}
thus we  obtain the following inequalities
\begin{align*}
\frac{(O-E_B)}{\sqrt{V_A}} &\leq \frac{(O-E_A)}{\sqrt{V_A}} < \frac{(O-E_A)}{\sqrt{V_B}}  \\
\frac{(O-E_B)}{\sqrt{V_A}} &< \frac{(O-E_B)}{\sqrt{V_B}} \leq  \frac{(O-E_A)}{\sqrt{V_B}} 
\end{align*}
Then the proof follows the same argument given in the appendix, and indeed $Z_B \leq Z_A$. \\

\item[-] If   $(O-E_B)<0$, then
\begin{align*}
\frac{1}{\sqrt{V_B}}&>\frac{1}{\sqrt{V_A}}\\
\frac{(O-E_B)}{\sqrt{V_B}}&<\frac{(O-E_B)}{\sqrt{V_A}}\leq\frac{(O-E_A)}{\sqrt{V_A}}
\end{align*}
thus $Z_B \leq Z_A$. %Note this include the case when $(O-E_B)$ is negative and $(O-E_A)$ is  negative.\\

 \end{itemize}

\item[(iib)]  If $Y_1(\tau_{k_j-1})-u_1<Y_0(\tau_{k_j-1})-u_0$ then $V_B>V_A$, i.e. $\frac{1}{V_B}<\frac{1}{V_A}$.\\
 \begin{itemize}
\item[-] If $(O-E_B)>0$ then $(O-E_A)$ also has to be positive.\\
 In this case, we  multiply $(O-E_B) \leq (O-E_A)$ by $\frac{1}{\sqrt{V_B}}<\frac{1}{\sqrt{V_A}}$, to obtain that $Z_B \leq Z_A$.\\
 
 \item[-] If $(O-E_A)>0$, then 
 $\frac{(O-E_A)}{\sqrt{V_B}}<\frac{(O-E_A)}{\sqrt{V_A}}$,\\
 and if we divide  $(O-E_B) \leq (O-E_A)$ by $\sqrt{V_B}$ then we have
 $\frac{(O-E_B)}{\sqrt{V_B}} \leq \frac{(O-E_A)}{\sqrt{V_B}} $ and therefore  we have  $Z_B \leq Z_A$.
Note this includes the case when $(O-E_B)$ is negative but $(O-E_A)$ is  positive.\\

\item[-] If both  $(O-E_B)$ and $(O-E_A)$ are negative\\

We multiply $(O-E_B) \leq (O-E_A)$ by $\frac{1}{\sqrt{V_B}}$ and  by $\frac{1}{\sqrt{V_A}}$ we have
\begin{align*}
\frac{(O-E_B)}{\sqrt{V_A}} &\leq \frac{(O-E_A)}{\sqrt{V_A}}   \\
 \frac{(O-E_B)}{\sqrt{V_B}} &\leq  \frac{(O-E_A)}{\sqrt{V_B}} 
\end{align*}
Now we multiply $\frac{1}{\sqrt{V_A}}>\frac{1}{\sqrt{V_B}}$ by $(O-E_A)$ and  by $(O-E_B)$ we have
\begin{align*}
 \frac{(O-E_A)}{\sqrt{V_A}} &< \frac{(O-E_A)}{\sqrt{V_B}}  \\
\frac{(O-E_B)}{\sqrt{V_A}} &< \frac{(O-E_B)}{\sqrt{V_B}}
\end{align*}
thus we  obtain the following inequalities
\begin{align*}
&\frac{(O-E_B)}{\sqrt{V_A}} \leq \frac{(O-E_A)}{\sqrt{V_A}} < \frac{(O-E_A)}{\sqrt{V_B}} \\
&\frac{(O-E_B)}{\sqrt{V_A}} < \frac{(O-E_B)}{\sqrt{V_B}} \leq  \frac{(O-E_A)}{\sqrt{V_B}}  
\end{align*}
Then the proof follows the same argument given in the appendix, and indeed $Z_B \leq Z_A$.\\

 \end{itemize}
 
 \end{itemize}

 \end{itemize}

\noindent {\bf Scenario 3 (S3): when $x_{j_i}$ is a censoring time and  $ y_j$ is a failure time }\\
This third scenario is illustrated in Figure \ref{plot2}.

\begin{figure}[ht]
\begin{center}
\resizebox{!}{6cm}{%
\begin{tikzpicture}
%bottom right
\draw[step=1cm,very thin] (5,0) grid (7,2);
\node at (6,2.3) {$\tau_{k_j}$};
\node at (9.2,1.5) {$Y_0(\tau_{k_j})=Y_0(\tau_{k_j-1})-u_0$};
\node at (9.5,.5) {$Y_1(\tau_{k_j})=Y_1(\tau_{k_j-1})-u_1-1$};
\node at (9.3,-.5) {$Y(\tau_{k_j})=Y(\tau_{k_j-1})-u-1$};

% bottom left

\draw[step=1cm,very thin] (0,0) grid (2,2);
\node at (1,2.3) {$\tau_{k_j-1}=y_j$};
\node at (2.9,1.5) {$Y_0(\tau_{k_j-1})$};
\node at (2.9, .5) {$Y_1(\tau_{k_j-1})$};
\node at (2.9,-.5) {$Y(\tau_{k_j-1})$};

% top left
\draw[step=1cm,very thin] (0,4) grid (2,6);
\node at (1,6.3) {$\tau_{k_j-1}$};
\node at (2.9,5.5) {$Y_0(\tau_{k_j-1})$};
\node at (2.9, 4.5) {$Y_1(\tau_{k_j-1})$};
\node at (2.9,3.5) {$Y(\tau_{k_j-1})$};

% top right
\draw[step=1cm,very thin] (5,4) grid (7,6);
\node at (6,6.3) {$\tau_{k_j}=y_j$};
\node at (10.2,5.5) {$Y_0(\tau_{k_j})=Y_0(\tau_{k_j-1})-u_0-\delta_0(\tau_{k_j-1})$};
\node at (10.2,4.5) {$Y_1(\tau_{k_j})=Y_1(\tau_{k_j-1})-u_1-\delta_1(\tau_{k_j-1})$};
\node at (9.3,3.5) {$Y(\tau_{k_j})=Y(\tau_{k_j-1})-u-1$};

\node at (-2,1) {After};
\node at (-.5,1.5) {$G_0$};
\node at (-.5, .5) {$G_1$};

\node at (-2,5) {Before};
\node at (-.5,5.5) {$G_0$};
\node at (-.5, 4.5) {$G_1$};

\end{tikzpicture}
}
\caption{Scenario S3}
\label{plot2}
\end{center}
\end{figure}
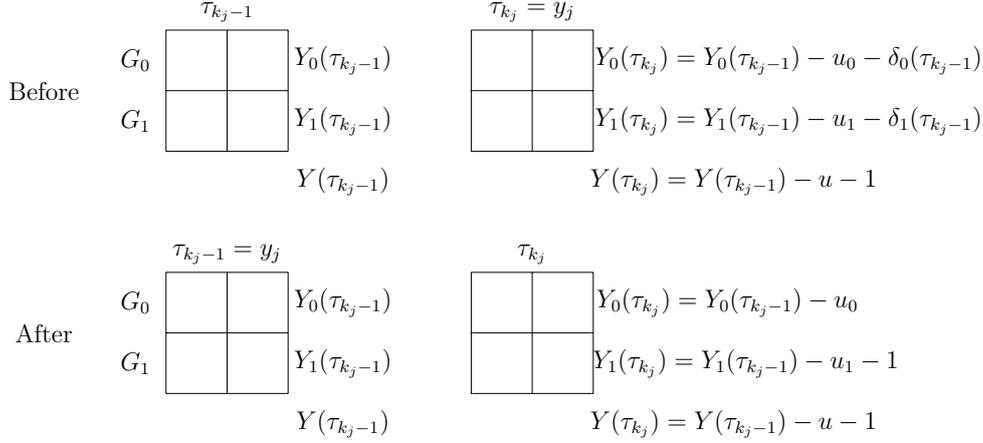

Similarly we  calculate the expected value and the variance before and after the swap as follows:
\begin{align*}
E_B&= \Sigma_{j\neq k_j} E_j+\left[\hat{S}(\tau_j)\right]^{\rho} \frac{Y_1(\tau_{k_j-1})-u_1-\delta_1(\tau_{k_j-1})}{Y(\tau_{k_j-1})-u-1} \\% before
E_A&= \Sigma_{j\neq k_j} E_j+\left[\hat{S}(\tau_j)\right]^{\rho} \frac{Y_1(\tau_{k_j-1})-u_1-1}{Y(\tau_{k_j-1})-u-1}  % after
\end{align*}
Clearly $E_B \geq E_A$, thus $(O-E_B) \leq (O-E_A)$.  And the variances are
\begin{align*}
V_B&= \Sigma_{j\neq k_j} V_j+\left[\hat{S}(\tau_j)\right]^{2\rho} \frac{(Y_0(\tau_{k_j-1})-u_0-\delta_0(\tau_{k_j-1}))(Y_1(\tau_{k_j-1})-u_1-\delta_1(\tau_{k_j-1}))}{(Y(\tau_{k_j-1})-u-1)^2} \\% before
V_A&= \Sigma_{j\neq k_j} V_j+\left[\hat{S}(\tau_j)\right]^{2\rho} \frac{(Y_0(\tau_{k_j-1})-u_0)(Y_1(\tau_{k_j-1})-u_1-1)}{(Y(\tau_{k_j-1})-u-1)^2}  % after
\end{align*}

 \begin{itemize}
\item[(i)] If $\delta_0(\tau_{k_j-1})=0$ then $\delta_1(\tau_{k_j-1})=1$ and $V_B=V_A$. Thus $Z_B \leq Z_A$, as from above $(O-E_B) \leq (O-E_A)$.\\

\item[(ii)] If $\delta_0(\tau_{k_j-1})=1$ then $\delta_1(\tau_{k_j-1})=0$, 
\[V_B= \Sigma_{j\neq k_j} V_j+ \left[\hat{S}(\tau_j)\right]^{2\rho}\frac{(Y_0(\tau_{k_j-1})-u_0-1)(Y_1(\tau_{k_j-1})-u_1)}{(Y(\tau_{k_j-1})-u-1)^2}\] % before

 then we have two sub-cases:
 \begin{itemize}
\item[(iia)] if $Y_1(\tau_{k_j-1})-u_1>Y_0(\tau_{k_j-1})-u_0$ then $V_B<V_A$,
\item[(iib)]  if $Y_1(\tau_{k_j-1})-u_1<Y_0(\tau_{k_j-1})-u_0$ then $V_B>V_A$.\\
  \end{itemize}
    \end{itemize}
The proof for both cases (iia) and (iib) are similar to scenario S2. \\

  %%%%%%%%%%%%%%%
 
%%%%%%%%%

\noindent {\bf Scenario 4 (S4):  when both $x_{j_i}$ and $ y_j$ are failure times}\\
This final scenario is illustrated in Figure \ref{plot3}.

\begin{figure}[ht]
\begin{center}
\resizebox{!}{6cm}{%
\begin{tikzpicture}
%bottom right
\draw[step=1cm,very thin] (5,0) grid (7,2);
\node at (6,2.3) {$\tau_{k_j}=x_{j_i}$};
\node at (8.8,1.5) {$Y_0(\tau_{k_j})=Y_0(\tau_{k_j-1})$};
\node at (9.1,.5) {$Y_1(\tau_{k_j})=Y_1(\tau_{k_j-1})-1$};
\node at (9,-.5) {$Y(\tau_{k_j})=Y(\tau_{k_j-1})-1$};

% bottom left

\draw[step=1cm,very thin] (0,0) grid (2,2);
\node at (1,2.3) {$\tau_{k_j-1}=y_{j}$};
\node at (2.9,1.5) {$Y_0(\tau_{k_j-1})$};
\node at (2.9, .5) {$Y_1(\tau_{k_j-1})$};
\node at (2.9,-.5) {$Y(\tau_{k_j-1})$};

% top left
\draw[step=1cm,very thin] (0,4) grid (2,6);
\node at (1,6.3) {$\tau_{k_j-1}=x_{j_i}$};
\node at (2.9,5.5) {$Y_0(\tau_{k_j-1})$};
\node at (2.9, 4.5) {$Y_1(\tau_{k_j-1})$};
\node at (2.9,3.5) {$Y(\tau_{k_j-1})$};

% top right
\draw[step=1cm,very thin] (5,4) grid (7,6);
\node at (6,6.3) {$\tau_{k_j}=y_{j}$};
\node at (9.1,5.5) {$Y_0(\tau_{k_j})=Y_0(\tau_{k_j-1})-1$};
\node at (8.8,4.5) {$Y_1(\tau_{k_j})=Y_1(\tau_{k_j-1})$};
\node at (9,3.5) {$Y(\tau_{k_j})=Y(\tau_{k_j-1})-1$};

\node at (-2,1) {After};
\node at (-.5,1.5) {$G_0$};
\node at (-.5, .5) {$G_1$};

\node at (-2,5) {Before};
\node at (-.5,5.5) {$G_0$};
\node at (-.5, 4.5) {$G_1$};

\end{tikzpicture}
}
\caption{Scenario S4}
\label{plot3}
\end{center}
\end{figure}
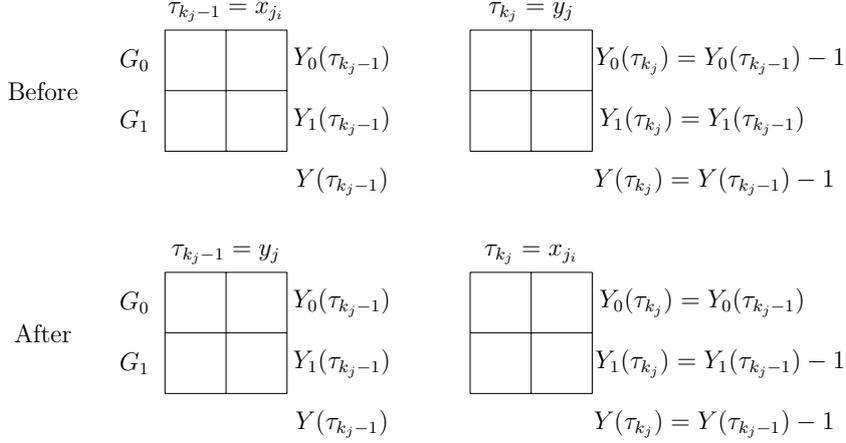

We  calculate the expected value  before and after the swap as follows:
\[E_B= \Sigma_{j\neq k_j} E_j+\left[\hat{S}(\tau_j)\right]^{\rho} \frac{Y_1(\tau_{k_j-1})}{Y(\tau_{k_j-1})-1}\] % before

\[E_A= \Sigma_{j\neq k_j} E_j+\left[\hat{S}(\tau_j)\right]^{\rho} \frac{Y_1(\tau_{k_j-1})-1}{Y(\tau_{k_j-1})-1} \] % after

Clearly $E_B \geq E_A$, thus $(O-E_B) \leq (O-E_A)$.  And the variances are

\[V_B= \Sigma_{j\neq k_j} V_j+\left[\hat{S}(\tau_j)\right]^{2\rho} \frac{(Y_0(\tau_{k_j-1})-1)Y_1(\tau_{k_j-1})}{(Y(\tau_{k_j-1})-1)^2}\] % before

\[V_A= \Sigma_{j\neq k_j} V_j+ \left[\hat{S}(\tau_j)\right]^{2\rho}\frac{Y_0(\tau_{k_j-1})(Y_1(\tau_{k_j-1})-1)}{(Y(\tau_{k_j-1})-1)^2} \] % after

  \begin{itemize}
\item[(i)] If $Y_0(\tau_{k_j-1})=Y_1(\tau_{k_j-1})$ then $V_B=V_A$. Thus $Z_B \leq Z_A$, as from above $(O-E_B) \leq (O-E_A)$.\\

\item[(ii)]  As we know that both $x_{j_i}$ and $ y_j$ are failure times, we have the two-sub cases:

  \begin{itemize}
\item[(iia)] If $Y_1(\tau_{k_j-1})>Y_0(\tau_{k_j-1})$ then $V_B<V_A$, thus $\frac{1}{V_B}>\frac{1}{V_A}$

  \begin{itemize}
\item[-] If $(O-E_B)>0$, then $(O-E_A)$ also has to be positive, we obtain similarly
 that
\begin{align*}
\frac{(O-E_B)}{\sqrt{V_A}} &\leq \frac{(O-E_A)}{\sqrt{V_A}} < \frac{(O-E_A)}{\sqrt{V_B}}  \\
\frac{(O-E_B)}{\sqrt{V_A}} &< \frac{(O-E_B)}{\sqrt{V_B}} \leq  \frac{(O-E_A)}{\sqrt{V_B}} 
\end{align*}
Then the proof follows the same argument given in the appendix, and indeed $Z_B \leq Z_A$.\\
 
\item[-] If   $(O-E_B)<0$, then
\begin{align*}
\frac{1}{\sqrt{V_B}}&>\frac{1}{\sqrt{V_A}}\\
\frac{(O-E_B)}{\sqrt{V_B}}&<\frac{(O-E_B)}{\sqrt{V_A}}\leq\frac{(O-E_A)}{\sqrt{V_A}}
\end{align*}
then $Z_B \leq Z_A$. %Note this include the case when $(O-E_B)$ is negative and $(O-E_A)$ is  negative.\\

 \end{itemize}

 \end{itemize}
\item[(iib)]  if $Y_1(\tau_{k_j-1})<Y_0(\tau_{k_j-1})$ then $V_B>V_A$, thus $\frac{1}{V_B}<\frac{1}{V_A}$.\\

\begin{itemize}
\item[-] If $(O-E_B)>0$ then $(O-E_A)$ also has to be positive.\\
 In this case, we can multiply $(O-E_B) \leq (O-E_A)$ by $\frac{1}{\sqrt{V_B}}<\frac{1}{\sqrt{V_A}}$, to obtain that $Z_B \leq Z_A$.\\
 
 \item[-] If $(O-E_A)>0$, then 
 $\frac{(O-E_A)}{\sqrt{V_B}}<\frac{(O-E_A)}{\sqrt{V_A}}$,\\
 and if we divide  $(O-E_B) \leq (O-E_A)$ by $\sqrt{V_B}$ then we have
 $\frac{(O-E_B)}{\sqrt{V_B}} \leq \frac{(O-E_A)}{\sqrt{V_B}} $ thus we obtain that $Z_B \leq Z_A$.
Note this include the case when $(O-E_B)$ is negative but $(O-E_A)$ is  positive.\\
 
\item[-] If both  $(O-E_B)$ and $(O-E_A)$ are negative, we obtain similarly
 that
\begin{align*}
&\frac{(O-E_B)}{\sqrt{V_A}} \leq \frac{(O-E_A)}{\sqrt{V_A}} < \frac{(O-E_A)}{\sqrt{V_B}} \\
&\frac{(O-E_B)}{\sqrt{V_A}} < \frac{(O-E_B)}{\sqrt{V_B}} \leq  \frac{(O-E_A)}{\sqrt{V_B}} 
\end{align*}
Then the proof follows the same argument given in the appendix, and indeed $Z_B \leq Z_A$.\\

\end{itemize}

\end{itemize}

%%%%%%%%%%%%
%%%%%%%%%%%%

\section{An example}\label{sec.example}
As illustrative example of the theorem proven in this paper, suppose that there are two groups with 5 observations each, where all observations from $G_0$ are smaller than all observations from group $G_1$. Let the censoring status for $G_0$  be (1, 0, 1, 1, 0) and for $G_1$ be  (1, 1, 0, 1, 0). For example we can choose group $G_0$ observations as 1, $2^+$,  3, 4, $5^+$  and group $G_1$ observations as 6,  7,  $8^+$, 9, $10^+$.  Figure \ref{fig.ex3}  shows the $z$-test values, obtained using Equation \eqref{eq.z}, for all 26 switches for  different values of $\rho$, $\rho=\{0,0.1,0.2,\ldots,1\}$. The $z$-values for all 26 switches and  for $\rho=0.5$ are given in Table \ref{tab.ex1}. The first thing one can observe here is that the $z$-test values are in   ascending order regardless of the values of $\rho$. For Table \ref{tab.ex1} and Figure \ref{fig.ex3} , we can also see that when switching between censored observations the $z$-values do not change, see e.g.\ switches (11,12), (14,15), (21,22) and (24,25).

\begin{figure}[ht]
\begin{center}
\includegraphics[width=0.98\textwidth] {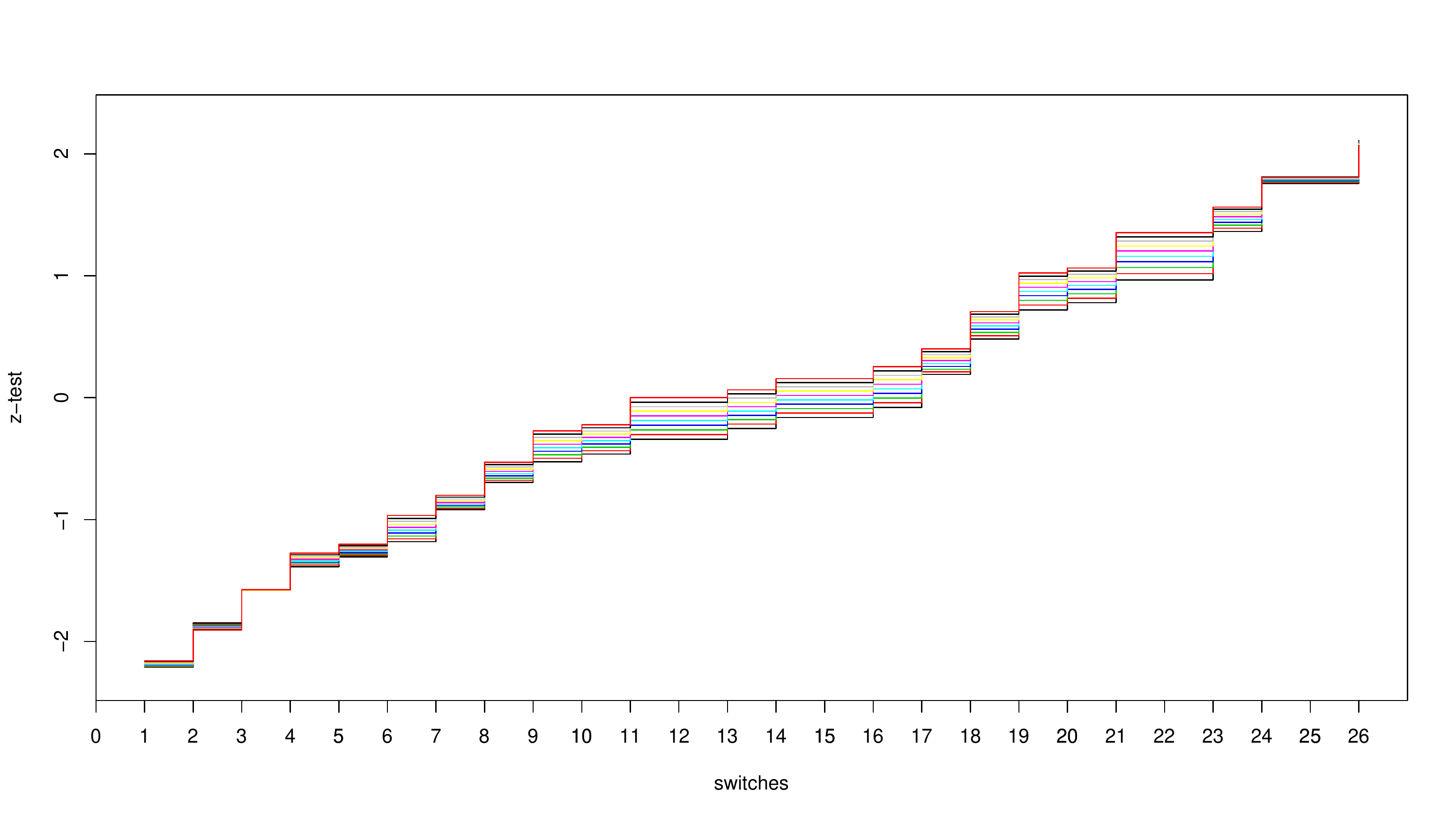}
\caption{$z$-test values for different values of $\rho$}
\label{fig.ex3}
\end{center}
\end{figure}

\begin{table}[ht]
\centering\footnotesize
\begin{tabular}{c|ccccc|ccccc|c}
  \hline
% & 1 & $2^+$ & 3 & 4 & $5^+$ & 6 & 7 & $8^+$ & 9 & $10^+$ & $z$-test \\ 
& \multicolumn{5}{c|}{$G_0$}&\multicolumn{5}{c|}{$G_1$}& $z$-test \\ 
  \hline
1 & 1 & $2^+$ & 3 & 4 & $5^+$ & 6 & 7 & $8^+$ & 9 & $10^+$ & -2.1901 \\ 
  2 & 1 & $2^+$ & 3 & 4 & 6 & $5^+$ & 7 & $8^+$ & 9 & $10^+$ & -1.8797 \\ 
  3 & 1 & $2^+$& 3 & 6 & 4 & $5^+$ & 7 & $8^+$ & 9 & $10^+$ & -1.5791 \\ 
  4 & 1 & $2^+$ & 6 & 3 & 4 & $5^+$ & 7 & $8^+$ & 9 & $10^+$ & -1.3374 \\ 
  5 & 1 & 6 & $2^+$ & 3 & 4 & $5^+$ & 7 & $8^+$ & 9 & $10^+$ & -1.2602 \\ 
  6 & 6 & 1 & $2^+$ & 3 & 4 & $5^+$ & 7 & $8^+$ & 9 & $10^+$ & -1.0872 \\ 
  7 & 6 & 1 & $2^+$ & 3 & 4 & 7 & $5^+$ & $8^+$ & 9 & $10^+$ & -0.8729 \\ 
  8 & 6 & 1 & $2^+$ & 3 & 7 & 4 & $5^+$ & $8^+$ & 9 & $10^+$ & -0.6236 \\ 
  9 & 6 & 1 & $2^+$ & 7 & 3 & 4 & $5^+$ & $8^+$ & 9 & $10^+$ & -0.4107 \\ 
  10 & 6 & 1 & 7 & $2^+$ & 3 & 4 & $5^+$ & $8^+$ & 9 & $10^+$ & -0.3533 \\ 
  11 & 6 & 7 & 1 & $2^+$ & 3 & 4 & $5^+$ & $8^+$ & 9 & $10^+$ & -0.1873 \\ 
  12 & 6 & 7 & 1 & $2^+$ & 3 & 4 & $8^+$ & $5^+$ & 9 & $10^+$ & -0.1873 \\ 
  13 & 6 & 7 & 1 & $2^+$ & 3 & $8^+$ & 4 & $5^+$ & 9 & $10^+$ & -0.1096 \\ 
  14 & 6 & 7 & 1 & $2^+$ & $8^+$ & 3 & 4 & $5^+$& 9 & $10^+$ & -0.0175 \\ 
  15 & 6 & 7 & 1 & $8^+$ & $2^+$ & 3 & 4 & $5^+$ & 9 & $10^+$ & -0.0175 \\ 
  16 & 6 & 7 & $8^+$ & 1 & $2^+$ & 3 & 4 & $5^+$ & 9 & $10^+$ & 0.0722 \\ 
  17 & 6 & 7 & $8^+$ & 1 & $2^+$ & 3 & 4 & 9 & $5^+$ & $10^+$ & 0.2798 \\ 
  18 & 6 & 7 & $8^+$ & 1 & $2^+$ & 3 & 9 & 4 & $5^+$ & $10^+$ & 0.5884 \\ 
  19 & 6 & 7 & $8^+$ & 1 & $2^+$ & 9 & 3 & 4 & $5^+$ & $10^+$ & 0.8718 \\ 
  20 & 6 & 7 & $8^+$ & 1 & 9 & $2^+$ & 3 & 4 & $5^+$ & $10^+$ & 0.9208 \\ 
  21 & 6 & 7 & $8^+$ & 9 & 1 & $2^+$ & 3 & 4 & $5^+$ & $10^+$ & 1.1602 \\ 
  22 & 6 & 7 & $8^+$ & 9 & 1 & $2^+$ & 3 & 4 & $10^+$ & $5^+$& 1.1602 \\ 
  23 & 6 & 7 & $8^+$ & 9 & 1 & $2^+$ & 3 & $10^+$ & 4 & $5^+$ & 1.4627 \\ 
  24 & 6 & 7 & $8^+$ & 9 & 1 & $2^+$ & $10^+$ & 3 & 4 & $5^+$ & 1.7874 \\ 
  25 & 6 & 7 & $8^+$ & 9 & 1 & $10^+$ & $2^+$ & 3 & 4 & $5^+$ & 1.7874 \\ 
  26 & 6 & 7 & $8^+$ & 9 & $10^+$ & 1 & $2^+$& 3 & 4 & $5^+$ & 2.0898 \\ 
   \hline
\end{tabular}
   \caption{$z$-test values and $\rho=0.5$}
\label{tab.ex1}
\end{table}

%%%%%%%%%%%%%%%%%%

\section{Concluding remarks} \label{sec.con}
\label{concl}

This paper studies the monotonicity of the $G^{\rho}$ class of weighted logrank tests  introduced by \cite{HF82}. 
We proved a convenient monotonicity property for the two-sample class of logrank tests. This property holds trivially
for the special case where there are no right censored observations (the Wilcoxon test), but, while intuitively quite
clear, its proof required care due to the right censoring affecting the data. One can utilise this property to derive
optimal bounds for the test statistic in case of imprecise data; such uses will be investigated for particular applications
and reported elsewhere. It is also worth investigation the construction of statistical tests for equality of survival functions
based on the number of switches, in a way that is similar to tests for perfect ranking in ranked set sampling presented by \citet{LI20081325}.  It is also of interest to study the generalization of the monotonicity property for such tests with
more than two groups of data, this  is left as a topic for future research.

%%%%%%%%%%%%%%%%%%
\section*{Appendix}
\begin{lemma}
Let $x$ and $y$ be any two real numbers in an interval  $[a,b]$, where $a<b$. Then  if $x-a<y-a$ and $b-y<b-x$ then $x<y$.
\end{lemma}
\begin{proof}
The setting is illustrated in the figure below.
%\begin{figure}[h]
\begin{center}
%\resizebox{!}{6cm}{%
\begin{tikzpicture}

\draw(0,0)--(10,0);

\foreach \x/\xtext in {0/$a$,2.5/$x$,7.5/$y$,10/$b$}
    \draw(\x,5pt)--(\x,-5pt) node[below] {\xtext};

\draw[decorate, decoration={brace, amplitude=5pt,raise=1pt}, yshift=2ex,blue]  (0,0) -- node[above=0.4ex] {$x-a$}  (2.5,0);
\draw[decorate, decoration={brace, amplitude=5pt,mirror,raise=1pt}, yshift=2ex,blue]  (10,0) -- node[above=0.4ex] {$b-x$}  (2.5,0);

\draw[decorate, decoration={brace, amplitude=5pt,raise=1pt}, yshift=2ex, red]  (0,0.8) -- node[above=0.4ex] {$y-a$}  (7.5,0.8);
\draw[decorate, decoration={brace, amplitude=5pt,mirror,raise=1pt}, yshift=2ex,red]  (10,0.8) -- node[above=0.4ex] {$b-y$}  (7.5,0.8);

\end{tikzpicture}

%}
%\caption{Simplified proof}
%\label{plot-appx}
\end{center}
%\end{figure}
As $y-a=(x-a)+(y-x)$ and $b-x=(b-y)+(y-x)$, then
in order for both inequalities to hold, the second term in the right hand side must be positive, i.e.\ $y-x>0$ thus $x<y$. %This is illustrated in    the figure below.
\end{proof}

We use the lemma above  to prove that  $\frac{O-E_A}{\sqrt{V_A}}>\frac{O-E_B}{\sqrt{V_B}}$. First we define the 4 differences $D_1$, $D_2$, $D_3$ and $D_4$, which is illustrated in  the figure below, as
%\begin{figure}[h]
\begin{center}
%\resizebox{!}{6cm}{%
\begin{tikzpicture}

\draw(0,0)--(10,0);

\foreach \x/\xtext in {0/$\frac{O-E_B}{\sqrt{V_A}}$,2.5/$\frac{O-E_B}{\sqrt{V_B}}$,7.5/$\frac{O-E_A}{\sqrt{V_A}}$,10/$\frac{O-E_A}{\sqrt{V_B}}$}
    \draw(\x,5pt)--(\x,-5pt) node[below] {\xtext};

\draw[decorate, decoration={brace, amplitude=5pt,raise=1pt}, yshift=2ex,blue]  (0,0) -- node[above=0.4ex] {$D_1$}  (2.5,0);
\draw[decorate, decoration={brace, amplitude=5pt,mirror,raise=1pt}, yshift=2ex,blue]  (10,0) -- node[above=0.4ex] {$D_2$}  (2.5,0);

\draw[decorate, decoration={brace, amplitude=5pt,raise=1pt}, yshift=2ex, red]  (0,0.8) -- node[above=0.4ex] {$D_3$}  (7.5,0.8);
\draw[decorate, decoration={brace, amplitude=5pt,mirror,raise=1pt}, yshift=2ex,red]  (10,0.8) -- node[above=0.4ex] {$D_4$}  (7.5,0.8);

\end{tikzpicture}

%}
%\caption{Simplified proof in the paper notation}
%\label{plot2-appx}
\end{center}
%\end{figure}
\begin{align*}
D_1&=\frac{O-E_B}{\sqrt{V_B}}-\frac{O-E_B}{\sqrt{V_A}}\\
D_2&=\frac{O-E_A}{\sqrt{V_B}}-\frac{O-E_B}{\sqrt{V_B}}\\
D_3&=\frac{O-E_A}{\sqrt{V_A}}-\frac{O-E_B}{\sqrt{V_A}}\\
D_4&=\frac{O-E_A}{\sqrt{V_B}}-\frac{O-E_A}{\sqrt{V_A}}
\end{align*}

In order for the inequality $\frac{O-E_A}{\sqrt{V_A}}>\frac{O-E_B}{\sqrt{V_B}}$ to hold,  both inequalities $D_3 > D_1$ and $D_2>D_4$ must be hold. We can express  $D_2$  and $D_3$  as
\begin{align*}
D_2&= D_4 + \left[\frac{O-E_A}{\sqrt{V_A}}-\frac{O-E_B}{\sqrt{V_B}}\right]\\
D_3&=D_1 +\left[\frac{O-E_A}{\sqrt{V_A}}-\frac{O-E_B}{\sqrt{V_B}}\right]
\end{align*}
so $\frac{O-E_A}{\sqrt{V_A}}-\frac{O-E_B}{\sqrt{V_B}}$ has to be positive, therefore $\frac{O-E_A}{\sqrt{V_A}}>\frac{O-E_B}{\sqrt{V_B}}$. %\qed

%\begin{align*}
%D_2=\frac{O-E_A}{\sqrt{V_B}}-\frac{O-E_B}{\sqrt{V_B}}&= \frac{O-E_A}{\sqrt{V_B}}-\frac{O-E_B}{\sqrt{V_B}}+\frac{O-E_A}{\sqrt{V_A}}-\frac{O-E_A}{\sqrt{V_A}}\\
%&= \left[\frac{O-E_A}{\sqrt{V_B}}-\frac{O-E_A}{\sqrt{V_A}}\right]+\left[\frac{O-E_A}{\sqrt{V_A}}-\frac{O-E_B}{\sqrt{V_B}}\right]\\
% &= D_4 + \left[\frac{O-E_A}{\sqrt{V_A}}-\frac{O-E_B}{\sqrt{V_B}}\right]
%\end{align*}

%\begin{align*}
%D_3=\frac{O-E_A}{\sqrt{V_A}}-\frac{O-E_B}{\sqrt{V_A}}&= \frac{O-E_A}{\sqrt{V_A}}-\frac{O-E_B}{\sqrt{V_A}}+\frac{O-E_B}{\sqrt{V_B}}-\frac{O-E_B}{\sqrt{V_B}}\\
%&= \left[\frac{O-E_B}{\sqrt{V_B}}-\frac{O-E_B}{\sqrt{V_A}}\right]+\left[\frac{O-E_A}{\sqrt{V_A}}-\frac{O-E_B}{\sqrt{V_B}}\right]\\
%&=D_1 +\left[\frac{O-E_A}{\sqrt{V_A}}-\frac{O-E_B}{\sqrt{V_B}}\right]
%\end{align*}

%-------------------------------------------------------------------------------
% Acknowledgments...
%-------------------------------------------------------------------------------
%\section*{Acknowledgement}
%\noindent ...

%-------------------------------------------------------------------------------
% References...
%-------------------------------------------------------------------------------

%\bibliographystyle{elsarticle-harv}
%\bibliography{ref-logrank}

\end{document}